\documentclass[manuscript]{acmart}

\AtBeginDocument{%
  \providecommand\BibTeX{{%
    \normalfont B\kern-0.5em{\scshape i\kern-0.25em b}\kern-0.8em\TeX}}}

\setcopyright{acmlicensed}
\acmJournal{TKDD}
\acmYear{2023} \acmVolume{1} \acmNumber{1} \acmArticle{1}
\acmMonth{1} \acmPrice{15.00}\acmDOI{10.1145/3583691}

\usepackage{amsmath}

\usepackage{amssymb}
\usepackage{bm}
\usepackage{nicefrac}
\usepackage{booktabs}
\usepackage{array}
\usepackage{multirow}
\usepackage{makecell}
\usepackage[procnumbered,ruled,vlined,linesnumbered]{algorithm2e}
\usepackage{stfloats}
\usepackage{graphicx}
\usepackage{subfigure}
\usepackage{enumerate}

\newtheorem{theorem}{Theorem}[section]

\newtheorem{lemma}[theorem]{Lemma}

\def\G{\mathcal{G}}

\newcommand{\Lap}{\ensuremath{\mathbf{L}}}
\newcommand{\A}{\ensuremath{\mathbf{A}}}
\newcommand{\D}{\ensuremath{\mathbf{D}}}

\newcommand{\removelatexerror}{\let\@latex@error\@gobble}

\newcommand\LL{\bm{\mathit{L}}}

\newcommand\uu{\boldsymbol{\mathit{u}}}

\newcommand\II{\boldsymbol{\mathit{I}}}

\newcommand{\one}{\mathbf{1}}

\DontPrintSemicolon
\SetKw{KwAnd}{and}
\SetFuncSty{textsc}
\SetKwInOut{Input}{Input\ \ \ \ }
\SetKwInOut{Output}{Output}

\usepackage{tabularx}
\usepackage{stfloats}

\begin{document}

\title{Optimal Scale-Free Small-World Graphs with Minimum Scaling of Cover Time}


\author{Wanyue Xu}
\email{xuwy@fudan.edu.cn}
\affiliation{%
  \institution{Fudan University}
  \country{China}
}
\author{Zhongzhi Zhang}
\email{zhangzz@fudan.edu.cn}
\affiliation{%
  \institution{Fudan University}
  \country{China}
}

\thanks{Both authors are with the Shanghai Key Laboratory of Intelligent Information Processing, School of Computer Science, Fudan University, Shanghai 200433, 
China. The work was supported by the Shanghai Municipal Science and Technology Major Project (No.  2018SHZDZX01), the National Natural Science Foundation of China (No. U20B2051), ZJLab, and Shanghai Center for Brain Science and Brain-Inspired Technology. Corresponding author: Zhongzhi Zhang.}

\renewcommand{\shortauthors}{Wanyue Xu and Zhongzhi Zhang}

\begin{abstract}
 The cover time of random walks on a graph has found wide practical applications in different fields of computer science, such as crawling and searching on the World Wide Web and query processing in sensor networks, with the application effects dependent on the behavior of cover time: the smaller the cover time, the better the application performance. It was proved that over all graphs with $N$ nodes, complete graphs have the minimum cover time $N\log N$. However, complete graphs cannot mimic real-world networks with small average degree and scale-free small-world properties, for which the cover time has not been examined carefully, and its behavior is still not well understood. In this paper, we first experimentally evaluate the cover time for various real-world networks with scale-free small-world properties, which scales as $N\log N$. To better understand the behavior of the cover time for real-world networks, we then study the cover time of three scale-free small-world model networks by using the connection between cover time and resistance diameter.  For all the three networks, their cover time also behaves as $N\log N$. This work indicates that sparse networks with scale-free and small-world topology are favorable architectures with optimal scaling of cover time. Our results deepen understanding the behavior of cover time in real-world networks with scale-free small-world structure, and have potential implications in the design of efficient algorithms related to cover time.
\end{abstract}


\maketitle

\section{Introduction}

As a paradigmatic dynamic process and a powerful analysis tool, random walks on a graph have attracted substantial attention from the scientific community~\cite{MaPoLa17}. One of the fundamental quantities associated with  random walks is the cover time~\cite{ChBeVo15}. For  random walks on a graph, the cover time is the expected steps a walker needs to visit every node on the graph. {The cover time has found a vast range of theoretical and practical applications in different areas. First, it can be used as a mechanism of crawl and search on the World Wide Web and peer-to-peer networks~\cite{CoFr02,CoFi03,AsLuPu01}, as well as food collection for grazing animals~\cite{ViBuHa99,BeLoMo11}. In addition, the cover time is much relevant to algorithmic design or analysis~\cite{JeSi96,KeDoGe03,GkMiSa04,Li12} in the contexts of replicated database maintenance and information spreading in graphs~\cite{FePeRa90,KaScSh00}. Finally, the cover time of a graph is also closely related to its combinatorial and algebraic properties such as the spectral gap and the conductance~\cite{BrKa89}. }

In view of its significant relevance, the cover time has been extensively studied~\cite{CoRa19,ClAlBu19,Vi20,ChDiLi21,BlJaRy22}. Since there exists no simple formula for computing cover time, a lot of authors focused on bounding this important quantity~\cite{Ma88,Al89,ChRaRu89,BrKa89,Zu90,AlKaLi79,Fe95,KaLiNi89}. Moreover, the cover time in various networks with particular topologies has received much interest. For instance, existing work has considered the cover time for many special graphs, including complete graphs~\cite{Fe95}, regular expander graphs~\cite{BrKa89}, path and cycle graphs~\cite{Lo96}, $k$-ary trees~\cite{Zu92}, torus graphs~\cite{Zu90,DePeRo04,BeKi17,Al91}, lollipop graphs~\cite{MoRa95}, bar-bell graphs~\cite{KaLiNi89},  and extended Sierpi{\'n}ski graphs~\cite{QiZh19}. These studies show that the cover time in different networks displays rich behaviors, which can scale with network size $N$ as $N\log N$, $ N\log^2 N $,  $N^2$, and $N^3$. Therefore, the network structure has a strong impact on the behavior of cover time.

In most application scenarios, it is ideal that the cover time is small. Previous work demonstrates that of all networks with the same number of nodes $N$, the complete graph is the unique graph with the exact minimum cover time $(N-1) \sum_{i=1}^{N-1}\frac{1}{i}$~\cite{Fe08}, the leading scaling of which is $N\log N$ for large $N$. This minimum scaling $N\log N$ can also be achieved in some other graphs, such as regular expander graphs~\cite{BrKa89} and $d$-dimensional torus graphs with $d \ge 3$~\cite{Zu90,DePeRo04,BeKi17,Al91}. However, these graphs with minimum scaling of cover time cannot well mimic real networks~\cite{Ne03}, most of which are simultaneously sparse, scale-free~\cite{BaAl99}, and small-world~\cite{WaSt98}. It has been established that the striking scale-free and small-world structure has a substantial effect on various dynamical processes running on networks, including games~\cite{SaSaPa08}, noisy consensus~\cite{YiZhLiCh15,YiZhSt19,XuWuZhZhKaCh22}, and disease spreading~\cite{ChWaWaLeFa08,VaOmKo09}, among others. It is also known that scale-free and small-world properties profoundly affect the behavior of many key quantities of random walks, i.e., relaxation time~\cite{NoRi04} and Kemeny constant~\cite{XuShZh20}. However, their impacts on the cover time is still not well understood. Particularly, it is largely unknown whether the minimum scaling $N\log N$ of cover time can be reached in scale-free  small-world networks. These motivate us to present an extensive study on the behavior of cover time for real and model networks, in order to explore the effects of scale-free small-world features on cover time.


The main contributions of this work are as follows. 
\begin{itemize}
    \item By using the connection between cover time and resistance distance, we study the cover time of sparse real networks with scale-free small-world properties, and show that the dominating scaling of their cover time behaves with network size $N$ as $N \log N$.
    
    \item We study  numerically the cover time on the Barab{\'a}si-Albert networks~\cite{BaAl99}, which display similar behavior as that of complete graphs.
    
    \item We  study analytically the cover time for two iteratively growing deterministic networks~\cite{DoMa05,DoGoMe02} with constant average degree, power-law degree distribution and small average shortest path distance.  Exploiting the decimation technique, we derive exactly the evolution relations of two-node resistance distance between two consecutive iterations for both networks. Based on the obtained relations, we provide the upper bounds of resistance diameters for the two graphs, both of which are small  constants, indicating that the cover time for both networks scales with network size $N$ as $N \log N$. 
    
    \item  We present a heuristic analysis, which shows that the scale-free small-world topology is responsible for the small cover time on all considered networks. 
\end{itemize}
 

\section{Preliminaries}

In this section, we introduce some basic concepts about a graph, the Laplacian matrix and its pseudoinverse, resistance distances of the corresponding electrical network, random walks and their key quantities, as well as some related work for the problem to be studied.

\subsection{Graph and Matrix Notation}

Let $\mathcal{G}=(\mathcal{V},\,\mathcal{E})$ denote a connected undirected unweighted graph with node set $\mathcal{V}$ and  edge set~$\mathcal{E} \subset \mathcal{V}\times \mathcal{V}$, the numbers of nodes and edges in which are $N = |\mathcal{V}|$ and  $E=|\mathcal{E}|$, respectively. Then, the total degree of all nodes is $2E$, and the average degree is $ d_{\rm{avg}}=(2E) /N$. Let $d_{\min }$  be  the minimum degree among all nodes. A graph is said to be simple if there is no loop or parallel edge. Throughout this paper, all considered graphs  are finite simple connected graphs, and the terms graph and network are used indistinctly. 
For a node $i \in \mathcal{V}$, let $\Delta_i = \{x|(x,i)\in \mathcal{E}\}$ denote the set of its neighbor nodes and let~$d_i = |\Delta_i|$ denote the degree of $i$. A graph is called scale-free~\cite{BaAl99} if its node degree $d$ follows a power-law distribution $P(d) \sim d^{-\gamma}$ with $\gamma > 0$. A graph is called small-world~\cite{WaSt98} if its average shortest path distance $l$ grows at most logarithmically with the number  of nodes $N$, that is, $l \leq \ln N$.

The $N$ nodes in graph $\mathcal{G}$ are labeled by $1,2,3,\ldots, N$, respectively. The adjacency relation between the $N$ nodes is encoded in its adjacency matrix $\A=(a_{ij})_{N \times N}$ of graph $\mathcal{G}$, where $a_{ij}=1$ if nodes $i$ and $j$ are directly connected by an edge in $\mathcal{E}$, and $a_{ij} = 0$~otherwise. Thus, the degree of node $i$ is  $d_i=\sum_{j=1}^{N} a_{ij}$. Let $\D$ denote the diagonal degree matrix of $\mathcal{G}$. The $i$th diagonal entry of $\D$ is $ d_i $, while all other entries are zeros.  Then, the Laplacian matrix $\LL$ of $\mathcal{G}$ is defined to be  $\LL=\D-\A$.

For a connected undirected unweighted graph  $\mathcal{G}=(\mathcal{V},\mathcal{E})$, its Laplacian matrix $\LL$ is symmetric and positive semidefinite. Then, all the eigenvalues of $\LL$ are non-negative, with a unique zero eigenvalue. Let $0=\lambda_1< \lambda_2 \leq \lambda_3\leq \dots\leq \lambda_{N-1} \leq \lambda_N$ denote the $N$ eigenvalues of matrix $\LL$, and let $\uu_k, k={1,2,\dots,N}$, denote their corresponding mutually orthogonal unit eigenvectors. Then, $\LL$ has a spectral decomposition of the form  $\LL=\sum_{k=1}^{N}\lambda_k \uu_k\uu_k^\top$. Thus,  the entry $\LL_{ij}$ at row $i$ and column $j$ of $\LL$ can be expressed as $\LL_{ij}=\sum_{k=2}^{N}\lambda_k u_{ki}u_{kj}$, where $u_{ki}$ is the $i$th component of vector $\uu_{k}$.

The Laplacian matrix $\LL$ is singular and cannot be inverted, since  $0$ is one of its eigenvalues. As a substitute for the inverse, we use the Moore-Penrose generalized inverse of $\LL$, which we  call pseudoinverse of $\LL$~\cite{BeGrTh74}. We use $\LL^{+}$ to denote its pseudoinverse, which can be written as
\begin{equation}\label{LPlus01}
\LL^{+}=\sum_{k=2}^{N}\frac{1}{\lambda_k}\uu_k\uu_k^{\top}.
\end{equation}

Let $\one$ and $\bm{J}$ denote, respectively, the vector and the matrix of appropriate dimensions with all entries being ones. Then,  $\bm{J}=\one\one^\top$. Let $\mathbf{0}$ and $\mathbf{O}$ denote, respectively, the zero vector and zero matrix of appropriate dimensions.  And let $\II$ be the identity matrix of  appropriate dimensions. The  symmetry of $\LL$ and $\LL^{+}$ implies that they share the same null space~\cite{BeGrTh74}. Since ${\LL} \one =\mathbf{0}$, then ${\LL}^{+} \one =\mathbf{0}$.  Furthermore, by using $\bm{J}=\one\one^\top$, we have  ${\LL} \bm{J} = \bm{J} {\LL}={\LL}^{+} \bm{J} = \bm{J} {\LL}^{+}= \mathbf{O}$.   Using the the above spectral decompositions of  $\LL$ and ${\LL}^{+}$, we have
\begin{equation*}
\left(\LL+\frac{1}{N} \bm{J}\right)\left(\LL^{+}+\frac{1}{N} \bm{J}\right)= \II.
\end{equation*}
Then, the pseudoinverse $\LL^{+}$ of $\LL$ can also be represented as
\begin{equation}\label{LPlus02} 
\LL^{+}=\left(\LL+\frac{1}{N} \bm{J}\right)^{-1}-\frac{1}{N} \bm{J}.
\end{equation}
This expression  is explicitly stated  in~\cite{GhBoSa08}  and is  implicitly applied in~\cite{XiGu03,BrFl05}.

\subsection{Resistance Distance of Electrical Network}

An electrical network associated with graph $\G$ is a network of resistances, where every edge in $\G$ is replaced by a unit resistance. In the case without incurring confusion, we also use $\G$ to denote the electrical network corresponding to graph $\G$.  The resistance distance between nodes $i$ and $j$ in $\G$, denoted by $\Omega_{ij}$, is defined as the potential difference between them when a unit current is injected at $i$ (or $j$) and extracted from $j$ (or $i$). It has been proved that the resistance distance is a metric~\cite{KlRa93}. Then, the resistance distance, also called  effective resistance, between any pair of nodes is symmetric, that is, $\Omega_{ji}=\Omega_{ij}$ holds for two arbitrary nodes $i$ and $j$. Let $R$ denote the resistance diameter of the electrical network, which  equals the maximum resistance distance among all pairs of nodes in the electrical network, in other words, 
\begin{align}
\label{R}
R=\max_{i\in \mathcal{V}, j \in \mathcal{V}}\Omega_{ij}\,.
\end{align}
It has been established~\cite{KlRa93} that  $\Omega_{ji}$ can be exactly represented in terms of the elements of  $\Lap^{+}$:
\begin{align}
\label{Resij}
\Omega_{ji}=\Omega_{ij}=\Lap^{+}_{ii}+\Lap^{+}_{jj}-2\Lap^{+}_{ij}\,.
\end{align}

Using the eigenvalue $\lambda_k$ and corresponding eigenvector $\uu_k$ of the Laplacian matrix for $ k={2,3, \ldots,N}$, the resistance distance $\Omega_{ij}$ between two nodes $i$ and $j$ can also be represented as~\cite{KlRa93}
\begin{equation}\label{rdis}
\Omega_{ij}= \sum_{k=2}^{N}\frac{1}{\lambda_k}(u_{ki}-u_{kj})^2.
\end{equation}
{The resistance distance has found applications in various fields, such as community detection~\cite{BeRaMi17, BeNaRaDh20}, link prediction~\cite{KuBeSa22}, as well as network centrality~\cite{LiZh18,LiPeShYiZh19,KuKuMi19,BeNaRaRa20}.} Moreover, the effective resistance of an electrical network has many interesting properties. For example,  effective resistances satisfy the following sum rule~\cite{Ch10}.
\begin{lemma}\label{sumRule}
	For any two different nodes $i$ and $j$ in  an  electrical network $\G=(\mathcal{V},\mathcal{E})$,  
	\begin{equation}
	d_i\Omega_{ij}+\sum_{k\in \Delta_{i}}(\Omega_{ik}-\Omega_{jk})=2\,.
	\end{equation}
\end{lemma}

\subsection{Random Walks on a Graph}

For a connected undirected network $\mathcal{G}=(\mathcal{V},\,\mathcal{E})$, consider unbiased discrete-time random walks on it. At each time step, the walker moves to a node uniformly chosen from the neighbors of  current location. Such a stochastic process is described by a Markov chain~\cite{KeSn76}, characterized by the transition matrix $\bf{T}=\D^{-1}\A$, with the $ij$th entry $t_{ij}= a_{ij}/d_i$ representing the probability of jumping to $j$ from $i$ in one time step. Let $\sigma_1$, $\sigma_2$, $\sigma_3$, $\ldots$, $\sigma_N$ be the $N$ eigenvalues of transition matrix $\bf{T}$, which can be ranked in decreasing order as $1=\sigma_1>\sigma_2 \geq \sigma_3 \geq \cdots \geq \sigma_N \geq -1$. The difference $1-\sigma_{2}$ between the largest eigenvalue and the second largest eigenvalue is called the spectral gap. 

For random walks on a graph, there are several fundamental quantities, including hitting time, commute time, cover time, and so on. For two nodes $u, v \in \mathcal{V}$, the hitting time $H_{uv}$ from $u$ to $v$ is defined as the expected number of steps a walker starting at $u$ requires to reach $v$ for the first time. Let $H_{\min}(\G)$ and $H_{\max}(\G)$ represent, respectively, the minimum and maximum hitting time among all pairs of nodes in graph $\G$.  For two nodes $u$ and $v$, their commute time $C_{uv}$ is the expected number of time steps a walker takes to go from $u$ to $v$ and back, that is, $C_{uv}= H_{uv}+ H_{vu}$. There is an elegant relation between commute time and effective resistance~\cite{ChRaRu89}: 
\begin{equation}
\label{com_res}
C_{uv}=2E\Omega_{uv}.
\end{equation}
Cover time is also an interesting quantity of random walks. For a node $u$, its cover time $C_u$ is defined as the expected number of steps necessary for a walker starting at $u$ to visit every node in $\mathcal{V}$. The cover time $C(\G)$ of the whole graph $\G$ is the maximum value of $C_u$ among all nodes in $\mathcal{V}$. 

\subsection{Related Work}

Due to the broad range of applications, the cover time has received considerable attention. Most previous work focus on two aspects: one is bounding the cover time, the other is uncovering the effects of network structure on the behavior of cover time.

Many techniques have been developed or used for  bounding the cover time~\cite{Ma88,Al89,ChRaRu89,BrKa89,Zu90,AlKaLi79,Fe95,KaLiNi89}. It was shown in the seminal work ~\cite{AlKaLi79} 
that for any connected graph $\G$, $C(\G)<2 N E$, which was later refined in~\cite{KaLiNi89} to obtain $C(\G)\le 4 N^{2} d_{\rm{avg}} / d_{\min }$. Thus, for regular graphs, $C(\G)\le 4 N^{2}$. In addition, the cover time of a graph can also be upper bounded in terms of the spectral gap of the transition matrix as $O\left(\frac{N \log N}{1-\sigma_{2}}\right)$\cite{MiPaSa03,MiPaSa06}.  According to the Matthews theorem~\cite{Ma88}, the cover time $C(\G)$ of graph $\G$ can be bounded by hitting times as:
\begin{align}
h_N H_{\min}(\G)\leq C(\G) \leq h_N H_{\max}(\G),
\end{align}
where $h_N$ is the $N$th harmonic number, given by  $h_N=\sum_{i=1}^{N} \frac{1}{i}  \approx \log N$. 
Using the connection between commute time and effective resistance in~\eqref{com_res}, it follows that~\cite{ChRaRu89},
\begin{equation}
\label{commuteR}
N\log N\le C(\G) \le 2 E\,R \log N,
\end{equation}
which means that for a sparse graph with constant average degree $d_{\rm{avg}}=(2E) /N$ and  constant resistance diameter $R$, its cover time achieves the minimum scaling $N \log N$.


In addition to bounding the cover time, a concerted effort has also been devoted to unveiling the influences of network structure on the behavior of cover time. Particularly, many groups have studied the cover time for  networks with different structural features. It was shown that in $N$-node networks with different  structures, their cover time often behaves differently with $N$. For the complete graph, its cover time  is $\Theta(N \log N)$~\cite{Fe95}, which is the possible minimum scaling among all graphs. For regular expander graphs~\cite{BrKa89} and those graphs with $d_{\min}\ge \lfloor{\frac{N}{2}}\rfloor$~\cite{ChRaRu89}, they also have expected cover time $\Theta(N \log N)$. For the path and cycle graphs~\cite{Lo96}, their cover time scales as the square of $N$. For the complete $k$-ary tree, the cover time is $\Theta(N\log^2 N)$~\cite{Zu92}. For $d$-dimensional torus graphs, the cover time is $\Theta(N \log^2 N)$ and $\Theta(N \log N)$ for $d=2$ and $d\ge3$, respectively~\cite{Zu90,DePeRo04,BeKi17,Al91}. Finally, for the lollipop graph~\cite{MoRa95} and the bar-bell graph~\cite{KaLiNi89}, their cover time is both $\Theta(N^3)$,  the possible maximum scaling among all graphs. 

As shown above, among all $N$-node graphs, the complete graph has the absolutely minimum cover time with scaling $N \log N$. A graph is called optimal if this minimum scaling $N \log N$ for cover time is attained.  In this sense, the complete graphs and regular expander graphs are optimal ones. However, the aforementioned graphs with minimum scaling of cover time cannot mimic realistic networks, which are often sparse, exhibiting the remarkable scale-free~\cite{BaAl99} small-world~\cite{WaSt98} properties  simultaneously. The scale-free property implies that the node degree obeys a power-law distribution $P(d)\sim d^{-\gamma}$. 
While the small-world property denotes that the average shortest path distance over all node pairs grows at most logarithmically with the network size $N$. In many practical applications, it is desirable that the cover time is as small as possible. Therefore, it is theoretically and practically interesting to find or design optimal sparse graphs having  scale-free small-world architecture, for which the minimal scaling $N\log N$ for the cover time can be reached.  

In the sequel, we will consider the cover time for scale-free small-world sparse graphs. We first study empirically the  cover time for some real-world scale-free small-world networks and show that their  cover time scales with network size $N$ as $N\log N$. We then evaluate the cover time  for  stochastic  Barab{\'a}si-Albert graphs~\cite{BaAl99} and two deterministically growing  scale-free small-world sparse graphs~\cite{DoMa05,DoGoMe02}, for all of which their cover time behaves with $N$ as $N\log N$. In this context, scale-free small-world networks have optimal structure with minimal scaling of cover time. 

\section{Cover Time in Realistic  Scale-Free Networks}

For random walks in a general graph, it is challenging, even impossible, to exactly evaluate the cover time, even if the structure of the graph is known beforehand. In this paper, we focus on the behavior of cover time. For this purpose, we use~\eqref{commuteR} to provide lower and upper bounds of cover time in different scale-free small-world networks, instead of computing the exact values. Our main goal is to unravel the influence of the scale-free small-world topology on the dominating scaling of cover time.

\subsection{Datasets of Real Networks}

All the datasets used in our experiment are chosen from the Koblenz Network Collection~\cite{Ku13} and Network Repository~\cite{RyNe15}. The collected networks are scale-free small-world, and are highly representative, spanning different fields such as social science, life science, and information science. Since we are only concerned with connected, undirected, unweighted simple graphs, we perform some preprocessings for the studied networks. For those directed or weighted networks, we convert directed or weighted edges to undirected and unweighted ones. And for each network, we only keep the largest connected component (LCC), deleting other small components and eliminating self-loops in the LCC. That is to say, we only study the cover time of the LCCs of resultant networks. In Table~\ref{RealNetwork}, we report related information of the  considered real-world networks, which  are listed in ascending order of the node number.

After preprocessing, we calculate the upper bound of the cover time for each network by using~\eqref{commuteR}. To this end, we numerically determine the resistance diameter corresponding to the LCC of each graph by using~\eqref{Resij} and~\eqref{R}. {Moreover, we simulate the actual cover time on the LCCs of these real-world networks.  The results are listed in the last three columns of Table~\ref{RealNetwork}.} 
From Table~\ref{RealNetwork}, we observe that for all studied the real scale-free networks,
their resistance diameter is significantly small in general. Actually, as found for various other properties (e.g., clustering coefficient~\cite{WaSt98}), the resistance diameter $R$ of real networks is not very sensitive  to the number of nodes $N$, but tends to very small constants. Thus, by \eqref{commuteR} the scaling of upper bound for the cover time is $ N\log  N$. {This is also confirmed in the last column  of Table~\ref{RealNetwork}, which shows that for each of these networks the ratio of its cover time and $ N\log  N$ is small.}  On the other hand, for an arbitrary graph, the lower bound of cover time is at most $ N\log  N$. We then conclude that for all the studied real networks with node number $N$, their cover time behaves with $N$ as $ N\log  N$, which is similar to that found for complete graphs.

\begin{table}
	\centering
	\tabcolsep=2.5pt
	\fontsize{8}{9}\selectfont
	\caption{ Statistics of some datasets and their resistance diameters. For each network $\G$, we indicate the number of nodes $N$, and the number of edges $E$, the average degree $d_{\rm{avg}}$, the  power-law exponent $\gamma$, the maximum for all pairs of shortest path lengths $\Delta$, and the resistance diameter $R$, {the cover time $C(\G)$, and the ratio $\frac{C(\G)}{N\log N}$} for its largest connected component.} 
	\label{RealNetwork}
	\resizebox{0.5\linewidth}{!}{
		\begin{tabular}{lrrrrrrrrr}
			\toprule
			Network & & $N$ & $E$ & $d_{\rm{avg}}$ & $\gamma$ & $\Delta$ & $R$ & {$C(\G)$} & {$\frac{C(\G)}{N\log N}$}\\
			\midrule
			karate & & 34 & 78 & 4.59 & 2.16 & 5 & 1.83 &	328	&	2.735\\
			windsurfers & & 43 & 336 & 15.63 & 4.00 & 3 & 0.34 &	229	&	1.415\\
			lesmis & & 77 & 254 & 6.60 & 1.52 & 5 & 3.62 &	1023	&	3.058\\
			adjnoun & & 112 & 425 & 7.59 & 3.62 & 5 & 2.59 &	2747	&	5.198\\
			bio-celegansneural & & 297 & 2148 & 14.46 & 3.34 & 5 & 2.08 &	32134	&	19.002\\
			bio-celegans & & 453 & 2025 & 8.94 & 2.63 & 7 & 4.59 &	12069	&	4.356\\
			ia-crime-moreno & & 829 & 1473 & 3.55 & 3.31 & 10 & 4.33 &	28242	&	5.069\\
			soc-wiki-Vote & & 889 & 2914 & 6.56 & 3.40 & 13 & 7.63 &	48986	&	8.115\\
			socfb-Reed98 & & 962 & 18812 & 39.11 & 2.80 & 6 & 3.02 &	147888	&	22.380\\
			soc-hamsterster & & 2000 & 16097 & 16.10 & 2.42 & 10 & 5.61 &	247492	&	16.280\\
			socfb-USFCA72 & & 2672 & 65244 & 48.84 & 2.50 & 7 & 3.63 &	704113	&	33.396\\
			socfb-nips-ego & & 2888 & 2981 & 2.06 & 4.52 & 9 & 6.35 &	66389	&	2.884\\
			bio-grid-worm & & 3343 & 6437 & 3.85 & 2.39 & 13 & 9.62 &	163195	&	6.015\\
			facebooknips & & 4039 & 88234 & 43.69 & 2.25 & 8 & 2.88 &	1121921	&	33.451\\
			ca-Erdos992 & & 4991 & 7428 & 2.98 & 2.18 & 14 & 7.61 &	185490	&	4.364\\
			routeviews & & 6474 & 12572 & 3.88 & 2.07 & 9 & 5.91 &	243710	&	4.289\\
			ia-reality & & 6809 & 7680 & 2.26 & 3.38 & 8 & 4.08 &	147961	&	2.462\\
			fb-pages-government & & 7057 & 89429 & 25.34 & 2.85 & 10 & 6.21 &	1837993	&	29.390\\
			soc-wiki-elec & & 7066 & 100727 & 28.51 & 1.42 & 7 & 4.30 &	2327121	&	37.158\\
			bio-dmela & & 7393 & 25569 & 6.92 & 3.39 & 11 & 8.10 &	861340	&	13.078\\
			soc-Blogcatalog-ASU & & 10312 & 333983 & 64.78 & 2.01 & 5 & 2.38 &	43014309	&	42.521\\
			ca-HepPh & & 11204 & 117619 & 21.00 & 2.09 & 13 & 6.75 &	1860637	&	17.810\\
			soc-Anybeat & & 12645 & 49132 & 7.77 & 1.75 & 10 & 7.21 &	1600433	&	13.400\\
			ca-AstroPh & & 17903 & 196972 & 22.00 & 2.87 & 14 & 6.00 &	3825904	&	21.822\\
			cond-mat & & 21363 & 91286 & 8.55 & 3.35 & 15 & 8.75 &	2262736	&	10.624\\
			tech-internet-as & & 40164 & 85123 & 4.24 & 2.09 & 11 & 6.55 &	2872041	&	6.745\\
			ego-gplus & & 23613 & 39182 & 3.32 & 2.62 & 8 & 5.51 &	796500	&	3.349\\
			\bottomrule
		\end{tabular}
	}
\end{table}

In order to deepen the understanding of the  behavior of cover time for realistic networks, in the following sections, we will determine analytically or numerically the scaling of cover time in three scale-free model networks, including Barab{\'a}si-Albert networks, Apollonian networks, and pseudofractal scale-free webs, and show that for the three networks, their cover time is $ \Theta(N\log  N)$. Thus, the $ N\log  N$ growth of cover time with network size $N$ in scale-free networks is universal.
\section{Cover time in Barab\'{a}si-Albert Networks}
In this section, we study the cover time for the popular Barab\'{a}si-Albert (BA) networks~\cite{BaAl99}, which capture the generating mechanisms for many real-world scale-free networks.
 
As an important scale-free network model, the BA networks~\cite{BaAl99} are  generated as follows. Initially, the network is a small connected graph, containing $m_0 \geq m$ nodes, with $m \geq 1$. At every time step, a new node is created and linked to $m$ different old nodes, with the probability of an old node $i$ connecting the new node being proportional to the degree $d_i$ of node $i$. Repeating these two procedures of growth and preferential attachment $g$ times, we obtain the BA network with $N=m_0+g$ nodes. When $N$ is sufficiently large, the average node degree $d_{\rm{avg}}$ of BA networks is approximately equal to $2m$. The node degree of the BA networks obeys a power-law distribution $P(d) \sim d^{-3}$. In addition to the scale-free property, the BA networks are small-world, the average shortest path distance of which increases logarithmically with the node number $N$~\cite{ChLu02}. 

According to the above procedures, we generate BA networks with various numbers of nodes and average degrees. We then study the cover time for random walks on these networks. For this purpose, we numerically compute their resistance diameter. Figure \ref{FigConBA} reports the resistance diameter of these BA networks with different $m=2$, $3$, $4$. From this figure, we can see that for large $N$,  the resistance diameter of these networks is independent of the network size $N$, but approaches to a small $m$-dependent constant: the larger the $m$, the smaller the resistance diameter. Therefore, according to~\eqref{commuteR}, the leading scaling $N \log N$ of cover time seems to be universal for BA networks. 

\begin{figure}
	\begin{center}
		\includegraphics[width=0.6\linewidth]{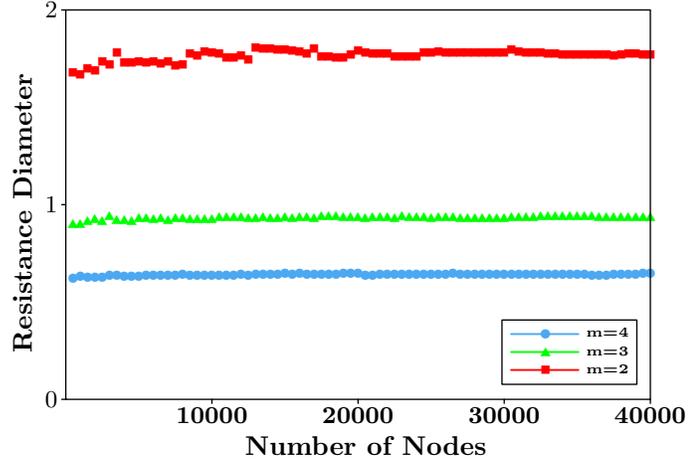}
	\end{center}
	\caption{Resistance diameters of the Barab{\'a}si-Albert networks with various $m$.}
	\label{FigConBA}
\end{figure}

\section{Cover time in Apollonian networks} \label{Apol}
This section is devoted to studying analytically the cover time of the Apollonian networks~\cite{DoMa05} with the common scale-free small-world characteristics as observed in real networks~\cite{Ne03}.

\subsection{Construction and Properties \label{modelA}}

The Apollonian networks, translated from the popular Apollonian packings, were introduced independently in~\cite{AnHeAnDa05} and in~\cite{DoMa05}. Although different initial constructions of Apollonian packing correspond to different networks, their structural and dynamical properties are similar. Here we consider the network version in~\cite{DoMa05}, which is defined in an iterative way~\cite{ZhRoZh06}. Let $\mathcal{A}_g=(\mathcal{V}_g, \mathcal{E}_g)$, $g \ge 0$, represent the Apollonian network after $g$ iterations. Initially, $\mathcal{A}_0$ is a tetrahedron composed of four triangles as shown in Figure~\ref{FigIni}(a). For $g\geq 1$, performing the following operations in Figure~\ref{FigIni}(b) on $\mathcal{A}_{g-1}$, one obtains $\mathcal{A}_g$: 
For each triangle in $\mathcal{A}_{g-1}$ that was generated at the $(g-1)$th iteration, create a new node and connect it to all the three nodes of this triangle.
Figure~\ref{netA} illustrates the evolution of the Apollonian network from $\mathcal{A}_0$ to $\mathcal{A}_2$.

Let $N_g= |\mathcal{V}_{g}|$ and $E_g = |\mathcal{E}_{g}|$ denote, respectively, the numbers of nodes and edges in $\mathcal{A}_{g}$. Let $\mathcal{W}_{g+1} = \mathcal{V}_{g+1} \backslash \mathcal{V}_{g}$ denote the set of new nodes generated at the $(g+1)$th iteration, and let~$W_g$ denote the cardinality of set $\mathcal{W}_g$. It is easy to derive that for all $g \ge 0$, $W_{g+1} = 4\times3^g$, $N_g=2\times3^{g}+2$, and  $E_g=6\times 3^g$. Therefore, the average node degree in network $\mathcal{A}_g$ is $2E_{g}/N_g$, which approaches to $6$ for large $g$, implying that the whole family of Apollonian networks is sparse. For a node $i$ in $\mathcal{A}_g$, let $d_i^{(g)}$ be its degree. Assume that node $i$ was generated at iteration $g_i$ ($g_i \geq 0$), then $d_i(g+1)=2\,d_i(g)=3\times 2^{g-g_i}$. 

The Apollonian networks exhibit the typical characteristics of many real natural and artificial  networks~\cite{DoMa05}. They are scale-free with the degree of their nodes obeying a power-law distribution $P(d)\sim d^{-(1 + \ln3/ \ln2)}$. Furthermore, they are small-world, since their diameter grows as a logarithmic function of the node number $N$~\cite{ZhRoZh06}. Finally, they are highly clustered with their average clustering coefficients tending to $0.8284$. Hence, Apollonian networks are an ideal model for mimicking real-life systems.
\begin{figure}
	\begin{center}
		\includegraphics[width=0.6\linewidth]{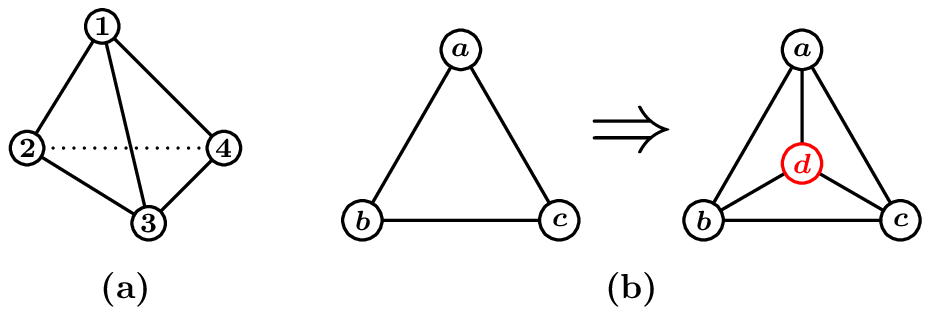}
	\end{center}
	\caption{(a) The initial Apollonian network. (b) Iterative construction approach of the Apollonian network. }
	\label{FigIni}
\end{figure}

\begin{figure}
	\begin{center}
		\includegraphics[width=0.5\linewidth]{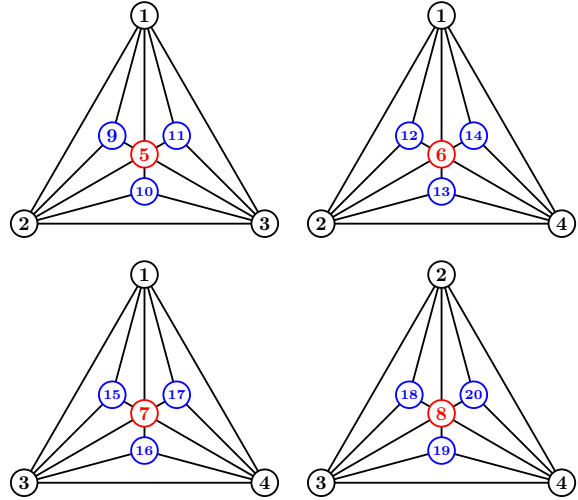}
		\caption{The Apollonian network $\mathcal{A}_{2}$. The evolution of the four faces in Figure~\ref{FigIni}(a) is  demonstrated seperately.}
		\label{netA}
	\end{center}
\end{figure}
\subsection{Recursive Relations for Matrices}
Denote $\A_g$ and $\D_g$ as the adjacency matrix and the diagonal degree matrix of network  $\mathcal{A}_g$, respectively. The element $\A_g(i,j)$ of matrix $\A_g$ at the $i$th row and the $j$th column is defined as follows: $\A_g(i,j)=1$ if nodes $i$ and $j$ are linked by an edge, or $\A_g(i,j)=0$ otherwise. The $i$th diagonal entry of degree matrix $\D_g$ is the degree $d_i^{(g)}$ of node $i$ in network $\mathcal{A}_g$. Then, the Laplacian matrix of network $\mathcal{A}_g$, denoted by $\LL_g$, is $\LL_g=\D_g - \A_g$. Below, we provide the recursion relations governing the evolution for the three matrices $\A_g$, $\D_g$, and $\LL_g$.

For the Apollonian network $\mathcal{A}_{g+1}$ after $g+1$ iterations, let $\alpha$ represent the set of old nodes already existing in the $g$th generation network  $\mathcal{A}_{g}$, and let $\beta$ be the set of new nodes in $\mathcal{W}_{g+1}$. Then, the adjacency matrix $\A_{g+1}$ of network $\mathcal{A}_{g+1}$ can be written in the following block form
\begin{align}
\A_{g+1}
= \left(
\begin{array}{cc}
\A_{g+1}^{\alpha,\alpha} & \A_{g+1}^{\alpha,\beta}\\
\A_{g+1}^{\beta,\alpha} & \A_{g+1}^{\beta,\beta}\nonumber\\
\end{array}
\right)=
\left(
\begin{array}{cc}
\A_g & \left(\A_{g+1}^{\beta,\alpha}\right)^\top\\
\A_{g+1}^{\beta,\alpha} & \bf{O} \nonumber
\end{array}
\right), 
\end{align}
where the submatrix $\A_{g+1}^{\alpha,\alpha}$ denotes the  adjacency relation between all pairs of old nodes $i$ and $j$, with $i,j\in \alpha$; submatrix $\A_{g+1}^{\alpha,\beta}$  represents the adjacency relation between all pairs of nodes $i$ and $j$, with $i \in \alpha$ and $j \in \beta$; similarly,  $\A_{g+1}^{\beta,\alpha}$ represents the  adjacency relation between those pairs of nodes $i$ and $j$, with  $i \in \beta$ and $j \in \alpha$; finally, $\A_{g+1}^{\beta,\beta}$ represents the adjacency relation between all pairs of new nodes $i$ and $j$, with both $i,j\in \beta$.  The second equality is accounted for as follows.  Since the adjacency relationship between old nodes keep unchangeable, $\A_{g+1}^{\alpha,\alpha} = \A_g$. On the other hand, since any pair of new nodes is not linked to each other, $\A_{g+1}^{\beta,\beta}$ is the  $W_{g+1} \times W_{g+1}$ zero matrix $\bf{O}$. Finally, by construction,  $\A_{g+1}^{\alpha,\beta}=\left(\A_{g+1}^{\beta,\alpha}\right)^\top$. 

The diagonal matrices $\D_{g+1}$ and $\D_g$ obey the relation 
\begin{align}
\D_{g+1} =
\left(
\begin{array}{cc}
\D_{g+1}^{\alpha,\alpha} & \bf{O}\\
\bf{O} & \D_{g+1}^{\beta,\beta}\nonumber\\
\end{array}
\right) =
\left(
\begin{array}{cc}
2\D_{g} & \bf{O}\\
\bf{O} & 3\II
\end{array}
\right),
\end{align}
which is obtained by using the following fact that in the evolution process of the Apollonian network from iteration $g$ to $g+1$, the degree of each old node in $\alpha$ doubles, and the degree of all newly added nodes in $\beta$ is 3.
Thus, the Laplacian matrix $\LL_g$ evolves in the following way 
\begin{align}
\LL_{g+1}
=\D_{g+1} - \A_{g+1} =
\left(
\begin{array}{cc}
2\D_g - \A_g & -\A_{g+1}^{\alpha,\beta}\\
-\A_{g+1}^{\beta,\alpha} & 3\II 
\end{array}
\right).
\end{align}
In this way, we have obtained recursive relations for relevant matrices, which are helpful for the following use.
\subsection{Relations between Effective Resistances}

As shown in~\eqref{commuteR}, the cover time of a connected graph $\G$ is closely related to its resistance diameter. Next, we study the scaling of the cover time for Apollonian networks by evaluating their resistance diameters. To achieve this goal, we first establish the evolution relationship governing  effective resistance between any two old nodes.

Instead of using the entries of pseudoinverse for the Laplacian matrix $\LL_g$, in what follows, we apply the entries of   $\{1\}-$inverse for $\LL_g$ to represent the effective resistance between any node pair. For a matrix $X$, matrix $M$ is called a $\{1\}-$inverse of  $X$, if and only if $XMX = X$~\cite{Ti94}. Note that for any matrix $X$, its pseudoinverse is also a $\{1\}-$inverse of $X$. Let $X^{\dag}$ represent one of the $\{1\}-$inverses of $X$. For any connected graph $\mathcal{G}$, the effective resistance between any two nodes can be represented in terms of the entries of any $\{1\}-$inverse of its Laplacian  matrix~\cite{Ba99}.
\begin{lemma}\label{efpro1}
	For a connected graph $\mathcal{G}=(\mathcal{V},\mathcal{E})$,
	let $\LL^{\dag}_{ij}$ denote the element at row $i$ and column $j$ of a $\{1\}-$inverse $\LL^{\dag}$ of its Laplacian matrix $\LL$. 
	Then, for any pair of nodes $i,j \in \mathcal{V}$,
	its effective resistance $\Omega_{ij}$ can be expressed in terms of the elements of $\LL^{\dag}$ as
	\begin{align}
	\Omega_{ij} =& \LL^{\dag}_{ii} + \LL^{\dag}_{jj} - \LL^{\dag}_{ij} - \LL^{\dag}_{ji}.
	\end{align}
\end{lemma}

The following lemma gives a $\{1\}-$inverse for a block square matrix~\cite{SuWaZh15}.
\begin{lemma}\label{LemmaBlock1inv}
	For a block matrix
	$X = \left(
	\begin{array}{cc}
	A & B\\
	B^{\top} & C \\
	\end{array}
	\right)
	$,
	where $C$ is an invertible matrix, if there exists a $\{1\}$-inverse $S^{\dag}$ for $S = A-BC^{-1}B^{\top}$,
	then
	\begin{align}
	X^{\dag} = \left(
	\begin{array}{cc}
	S^{\dag} & -S^{\dag}BC^{-1}\\
	-C^{-1}B^{T}S^{\dag} & C^{-1}B^TS^{\dag}BC^{-1} + C^{-1}
	\end{array}
	\right)
	\end{align}
	is a $\{1\}$-inverse of $X$.
\end{lemma}
\begin{lemma}\label{PFProA}
	For the Apollonian network $\mathcal{A}_{g+1}$ after $g+1$ $(g \ge 0)$ iterations,
	\begin{align}\label{PFAAtrans}
	\A_{g+1}^{\alpha,\beta}\A_{g+1}^{\beta,\alpha} = \D_g + 2\A_g.
	\end{align}
\end{lemma}
\begin{proof}
	To prove $\A_{g+1}^{\alpha,\beta}\A_{g+1}^{\beta,\alpha} =  \D_g + 2\A_g$, it is sufficient to show that the corresponding elements of the two matrices  $\A_{g+1}^{\alpha,\beta}\A_{g+1}^{\beta,\alpha} $ and $\D_g + 2\A_g$  are equal to each  other. For the convenience of description, define ${\bf{Q}}_g=\D_g + 2\A_g$ and ${\bf{Z}}_g=\A_{g+1}^{\alpha,\beta}\A_{g+1}^{\beta,\alpha}$. For matrix ${\bf{Q}}_g$, its diagonal entries are ${\bf{Q}}_g(i,i)=d_i^{(g)}$, and non-diagonal elements are ${\bf{Q}}_g(i,j)=2\A_g(i,j)$. Let ${\bf{Z}}_g(i,j)$ be the $(i,j)$th entry of matrix ${\bf{Z}}_g$. Below, we will show that  ${\bf{Z}}_g(i,j)={\bf{Q}}_g(i,j)$ for all $i,j \in \mathcal{V}$. 
	
	We first partition matrix  $\A_{g+1}^{\beta,\alpha}$ into $N_g$ column vectors as
	\begin{equation}
	\A_{g+1}^{\beta,\alpha} = (x_1,x_2,\ldots,x_{N_g}), \nonumber
	\end{equation}
	where each $x_i$ ($i=1,2,\ldots,N_g$) is a $W_{g+1}$-dimensional column vector, given by  $x_i=(x_{i,N_g+1},x_{i,N_g+2},\ldots, x_{i,N_{g+1}})^\top$ that describes the adjacency relation between node $i \in \alpha$ and all the nodes in $\beta$.  Because  $\A_{g+1}^{\alpha,\beta}=\left(\A_{g+1}^{\beta,\alpha}\right)^\top$, one has  $\A_{g+1}^{\alpha,\beta}=(x_1,x_2,\ldots,x_{N_g})^{\top}$ and 
	\begin{align}
	\A_{g+1}^{\alpha,\beta}\A_{g+1}^{\beta,\alpha} =& (x_1,x_2,\ldots,x_{N_g})^{\top} (x_1,x_2,\ldots,x_{N_g}) \nonumber\\
	=& (x_i^\top x_j )_{N_g \times N_g}.\nonumber
	\end{align}
	Next, we determine the entries ${\bf{Z}}_g(i,j)$ of $\A_{g+1}^{\alpha,\beta}\A_{g+1}^{\beta,\alpha}$ by distinguishing two cases: $i = j$ and $i \neq j$.
	
	For the first case $i = j$, the diagonal entry of ${\bf{Z}}_g$ is ${\bf{Z}}_g(i,i) = x_i^\top x_i $, which is in fact equal to the number of $i$'s new  neighboring nodes in $\beta$. Thus, ${\bf{Z}}_g(i,i) = d_i^{(g+1)}-d_i^{(g)} = d_i^{(g)}= {\bf{Q}}_g(i,i)$.
	
	For the second case $i \neq j$, ${\bf{Z}}_g(i,j)$ is  the $(i,j)$th non-diagonal entry of matrix ${\bf{Z}}_g$ which can be evaluated as
	\begin{align*}
	{\bf{Z}}_g(i,j) &= x_i^\top x_j =  \sum_{k \in \beta} (x_{i,k}x_{j,k})  \\
	&= \sum_{k \in \beta} (\A_{g+1}(i,k)\A_{g+1}(j,k))
	\\
	&= \sum_{\substack{\A_{g+1}(i,k) = 1 \\ \A_{g+1}(j,k) = 1 }}\A_{g}(i,j)
	\\
	&= 2\A_{g}(i,j) ={\bf{Q}}_g(i,j)\,,
	\end{align*}
	where we use the fact that each edge is included in two triangles created in iteration $g+1$. 
\end{proof}

For any pair of nodes $i$ and $j$ in network  $\mathcal{A}_g$, let~$\Omega_{ij}^{(g)}$~denote their resistance distance. When the network grows from iteration $g$ to $g+1$, the resistance distance  evolves according to the  relation given in the following lemma.
\begin{lemma}\label{V2V}
	Let $i,j \in \mathcal{V}_{g}$ be two old nodes in network $\mathcal{A}_{g+1}$ with $g\ge 0$. Then, the resistance distances $\Omega_{ij}^{(g)}$ and $\Omega_{ij}^{(g+1)}$ satisfy
	\begin{align}\label{oo}
	\Omega_{ij}^{(g+1)} = \frac{3}{5}\Omega_{ij}^{(g)}.
	\end{align}
\end{lemma}
\begin{proof}
	For network $\mathcal{A}_{g+1}$, any   $\{1\}-$inverse of its Laplacian $\LL_{g+1}^{\dag}$ can be written in block form as 
	\begin{align}
	\LL_{g+1}^{\dag} =
	\left(
	\begin{array}{cc}
	\left(\LL_{g+1}^{\alpha,\alpha}\right)^{\dag} & \left(\LL_{g+1}^{\alpha,\beta}\right)^{\dag} \\
	\left(\LL_{g+1}^{\beta,\alpha}\right)^{\dag}  & \left(\LL_{g+1}^{\beta,\beta}\right)^{\dag}  
	\end{array}
	\right).
	\end{align}
	By Lemmas~\ref{LemmaBlock1inv} and \ref{PFProA}, the submatrix $\left(\LL_{g+1}^{\alpha,\alpha}\right)^{\dag}$ can be expressed as 
	\begin{align}\label{oosubmat}
	\left(\LL_{g+1}^{\alpha,\alpha}\right)^{\dag}
	=& \left(2\D_g - \A_g -(-\A_{g+1}^{\alpha,\beta})(3\II)^{-1}(-\A_{g+1}^{\beta,\alpha}) \right)^{\dag} \nonumber\\
	=& \left((2\D_g-\A_g)-\frac{1}{3}(\D_g + 2 \A_g)\right)^{\dag} \nonumber\\
	=& \left(\frac{5}{3}\LL_g\right)^{\dag}
	= \frac{3}{5}\LL_g^{\dag}.
	\end{align}
	By Lemma~\ref{efpro1} and \eqref{oosubmat},
	for two old nodes $i,j \in \mathcal{V}_g$, their resistance distance between two consecutive iterations obeys
	\begin{align}
	\setlength{\parindent}{1.3em}\indent&\Omega_{ij}^{(g+1)} \nonumber\\
	=& \left(\LL_{g+1}^{\alpha,\alpha}\right)^{\dag}(i,i) + \left(\LL_{g+1}^{\alpha,\alpha}\right)^{\dag}(j,j)
	- \left(\LL_{g+1}^{\alpha,\alpha}\right)^{\dag}(i,j) - \left(\LL_{g+1}^{\alpha,\alpha}\right)^{\dag}(j,i)\nonumber\\
	=& \frac{3}{5} \left( \LL^{\dag}_{g}(i,i) + \LL^{\dag}_{g}(j,j)
	- \LL^{\dag}_{g}(i,j) - \LL^{\dag}_{g}(j,i) \right) \nonumber\\
	=& \frac{3}{5} \Omega_{ij}^{(g)},
	\end{align}
	which completes the proof.
\end{proof}


After obtaining the evolution relation of  resistance distance between any pair of old nodes in $\mathcal{A}_{g+1}$, we continue to show that the effective resistance between any other node pairs in $\mathcal{A}_{g+1}$ can be presented in terms of effective resistances for some pairs of old nodes in $\mathcal{A}_{g}$. 
To gain this goal, we introduce some additional  quantities.
For two node sets $F$ and $Y$, define
\begin{align}
\Omega_{F,Y}^{(g)} = \displaystyle \sum_{i \in F ,j \in Y} \Omega_{ij}^{(g)}.
\end{align}
For a node $i \in \mathcal{W}_{g+1}$ in $\mathcal{A}_{g+1}$, let $ {p,q,r} $ be its three neighbors, all of which are in $\mathcal{V}_g$ and constitute the set $\Delta_i= \{p,q,r\} $.  Define 
\begin{align}\label{comb4}
\Omega_{\Delta_i}^{(g)} = \Omega_{pq}^{(g)} + \Omega_{qr}^{(g)} + \Omega_{rp}^{(g)}.
\end{align}

\begin{lemma}\label{iDi}
	In the Apollonian network $\mathcal{A}_{g+1}$ with $g\ge0$, for any  new node $i \in \mathcal{W}_{g+1}$, the following relation holds: 
	\begin{align}
	\Omega_{i,\Delta_i}^{(g+1)} =& 1 + \frac{1}{3} \Omega_{\Delta_i}^{(g+1)}.
	\end{align}
\end{lemma}
\begin{proof}
	By Lemma~\ref{sumRule}, for~$i \in \mathcal{W}_{g+1}$ and its three neighboring nodes $p$, $q$, and $r$ forming set $\Delta_i= \{p,q,r\}$, one has
	\begin{align}
	3 \Omega_{ip}^{(g+1)} + \Omega_{i,\Delta_i}^{(g+1)} - \Omega_{p,\Delta_i}^{(g+1)} = 2, \nonumber \\
	3 \Omega_{iq}^{(g+1)} + \Omega_{i,\Delta_i}^{(g+1)} - \Omega_{q,\Delta_i}^{(g+1)} = 2, \nonumber 
	\end{align}
	and
	\begin{equation}
	3 \Omega_{ir}^{(g+1)} + \Omega_{i,\Delta_i}^{(g+1)} - \Omega_{r,\Delta_i}^{(g+1)} = 2.\nonumber
	\end{equation}
	Summing these three equations leads to
	\begin{align}
	3 \Omega_{i,\Delta_i}^{(g+1)} + 3\Omega_{i,\Delta_i}^{(g+1)} - \Omega_{\Delta_i,\Delta_i}^{(g+1)} = 6, 
	\end{align}
	which can be recast as 
	\begin{align}
	\Omega_{i,\Delta_i}^{(g+1)}
	=& 1 + \frac{1}{6}\Omega_{\Delta_i,\Delta_i}^{(g+1)} 
	= 1 + \frac{1}{3} \Omega_{\Delta_i}^{(g+1)}.
	\end{align}
	This completes the proof.
\end{proof}

\begin{lemma}\label{V2L}
	In the Apollonian network $\mathcal{A}_{g+1}$ ($g\ge0$), for any pair of nodes $i$ and $j$, with $i \in \mathcal{W}_{g+1}$ and $j \in \mathcal{V}_g$, the following relation holds: 
	\begin{align}
	\Omega_{ij}^{(g+1)} = \frac{1}{3} \left( 1 - \frac{1}{3}\Omega_{\Delta_i}^{(g+1)} + \Omega_{\Delta_i,j}^{(g+1)} \right).
	\end{align}
\end{lemma}
\begin{proof}
	According to Lemma~\ref{sumRule}, for~$i \in \mathcal{W}_{g+1}$ and $j \in \mathcal{V}_g$,  one obtains
	\begin{align}
	d_i^{(g+1)} \Omega_{ij}^{(g+1)} + \Omega_{\Delta_i,i}^{(g+1)} - \Omega_{\Delta_i,j}^{(g+1)} = 2. \nonumber
	\end{align}
	Considering $d_i^{(g+1)}=3$ and applying  Lemma~\ref{iDi} yield 
	\begin{align}
	\Omega_{ij}^{(g+1)}
	=& \frac{1}{3}\left(2 - \Omega_{\Delta_i,i}^{(g+1)} + \Omega_{\Delta_i,j}^{(g+1)} \right)\nonumber \\
	=& \frac{1}{3}\left(1 - \frac{1}{3} \Omega_{\Delta_i}^{(g+1)} + \Omega_{\Delta_i,j}^{(g+1)}\right),  \nonumber
	\end{align}
	as required.
\end{proof}

\begin{lemma}\label{newnodes}
	In the Apollonian network $\mathcal{A}_{g+1}$ ($g\ge0$), for any pair of different new nodes $i$ and $j$ in  $\mathcal{W}_{g+1}$, 
	\begin{align}
	\Omega_{ij}^{(g+1)} = \frac{2}{3} - \frac{1}{9}\left(\Omega_{\Delta_i}^{(g+1)} + \Omega_{\Delta_j}^{(g+1)}\right) +\frac{1}{9}\Omega_{\Delta_i,\Delta_j}^{(g+1)} .
	\end{align}
\end{lemma}
\begin{proof}
	For any two nodes $i$ and $j$ belonging to $\mathcal{W}_{g+1}$, by Lemma~\ref{sumRule}, one has 
	\begin{align}
	d_i^{(g+1)} \Omega_{ij}^{(g+1)} + \Omega_{\Delta_i,i}^{(g+1)} - \Omega_{\Delta_i,j}^{(g+1)} = 2. \nonumber
	\end{align}
	Using $d_i^{(g+1)}=3$, Lemma~\ref{iDi}, and Lemma~\ref{V2L}, it follows that 
	\begin{align}
	&\setlength{\parindent}{1.3em}\indent\Omega_{ij}^{(g+1)} = \frac{1}{3}\left(2 -  \Omega_{i,\Delta_i}^{(g+1)} + \Omega_{j,\Delta_i}^{(g+1)}\right) \nonumber\\
	&= \frac{1}{3}\left(2 -  \Omega_{i,\Delta_i}^{(g+1)} + \sum_{k \in \mathcal{N}(i)} \Omega_{j,k}^{(g+1)}\right)\nonumber \\
	&= \frac{1}{3}\left(2 -  \Omega_{i,\Delta_i}^{(g+1)} + \sum_{k \in \mathcal{N}(i)}\frac{1}{3} \left( 1 - \frac{1}{3}\Omega_{\Delta_j}^{(g+1)} + \Omega_{\Delta_j,k}^{(g+1)} \right) \right)\nonumber \\
	&= \frac{1}{3}\left(1 - \frac{1}{3} \Omega_{\Delta_i}^{(g+1)} +
	\left(1-\frac{1}{3}\Omega_{\Delta_j}^{(g+1)}+\frac{1}{3}\Omega_{\Delta_i,\Delta_j}^{(g+1)}\right) \right) \nonumber\\
	&= \frac{2}{3} - \frac{1}{9}\left(\Omega_{\Delta_i}^{(g+1)}+\Omega_{\Delta_j}^{(g+1)}\right)+\frac{1}{9}\Omega_{\Delta_i,\Delta_j}^{(g+1)},
	\end{align} 
	as claimed.
\end{proof}
\subsection{Scaling of the Cover Time} \label{Cover}
We are now in position to determine the leading scaling for the cover time $C(\mathcal{A}_{g})$ of the Apollonian network $\mathcal{A}_{g}$, by applying the connection between cover time and resistance diameter. Before doing so, we provide an upper bound for the resistance distance between any pair of nodes in $\mathcal{A}_{g}$. 
\begin{lemma}\label{efub}
	For any pair of nodes $i$ and $j$ in the Apollonian network $\mathcal{A}_{g}$ with $g \geq0$, their resistance distance satisfies $\Omega_{ij}^{(g)}\leq \frac{5}{3} $.
\end{lemma}
\begin{proof}
	We prove this lemma by induction.
	When $g = 0$, $\mathcal{A}_g$~is a tetrahedron, the result is true since $\Omega_{ij}^{(0)} = \frac{1}{2} \leq \frac{5}{3}$. 
	Suppose that the relation  holds  true for $g=n$, that is, $\Omega_{ij}^{(n)} \leq \frac{5}{3}$. We next prove the relation $\Omega_{ij}^{(g)} \leq \frac{5}{3}$ also holds for $g=n+1$. Note that for all pairs of nodes  $i$ and $j$ in network  $\mathcal{A}_{n+1}$, they can be categorized into three cases:  (i) $i \in \mathcal{V}_{n}$, $j \in \mathcal{V}_{n}$; (ii) $i \in \mathcal{W}_{n+1}$, $j \in \mathcal{V}_{n}$; and (iii) $i \in \mathcal{W}_{n+1}$, $j \in \mathcal{W}_{n+1}$.
	We next prove the relation $\Omega_{ij}^{(n+1)} \leq \frac{5}{3}$ by distinguishing these three cases. 
	
	For the first case of $i \in \mathcal{V}_{n}$ and $j \in \mathcal{V}_{n}$, by~\eqref{oo}, one has 
	\begin{equation}
	\Omega_{ij}^{(n+1)} = \frac{3}{5}\Omega_{ij}^{(n)} \leq \frac{5}{3}.
	\end{equation}
	
	For the second case of  $i \in \mathcal{W}_{n+1}$ and $j \in \mathcal{V}_{n}$, by Lemma~\ref{V2L} and induction assumption, one obtains 
	\begin{align}
	\Omega_{ij}^{(n+1)} =& \frac{1}{3} \left( 1 - \frac{1}{3}\Omega_{\Delta_i}^{(n+1)} + \Omega_{\Delta_i,j}^{(n+1)} \right) \nonumber \\
	\leq& \frac{1}{3} \left( 1 + \frac{3}{5}\Omega_{\Delta_i,j}^{(n)} \right) \nonumber \\
	\leq& \frac{1}{3} \left( 1 + \frac{3}{5}\times 3
	\times \frac{5}{3} \right) 
	\leq \frac{5}{3}.\nonumber
	\end{align}
	
	For the third case of $i \in \mathcal{W}_{n+1}$ and $j \in \mathcal{W}_{n+1}$, by Lemma~\ref{newnodes} and induction assumption, one gets 
	\begin{align}
	\Omega_{ij}^{(n+1)} =& \frac{2}{3} - \frac{1}{9}\left(\Omega_{\Delta_i}^{(n+1)}+\Omega_{\Delta_j}^{(n+1)}\right)+\frac{1}{9}\Omega_{\Delta_i,\Delta_j}^{(n+1)} \nonumber \\
	\leq&\frac{2}{3}+\frac{1}{9}\Omega_{\Delta_i,\Delta_j}^{(n+1)}\nonumber \\
	\leq&\frac{2}{3}+\frac{1}{9} \times \frac{3}{5}\Omega_{\Delta_i,\Delta_j}^{(n)}\nonumber \\
	\leq&\frac{2}{3}+\frac{1}{9} \times \frac{3}{5} \times 9 \times \frac{5}{3}
	\leq\frac{5}{3} .\nonumber 
	\end{align}
	This completes the proof.
\end{proof}

\begin{theorem}
	For $g \geq 0$,  the leading scaling of cover time $C(\mathcal{A}_g) $ for the Apollonian network  $\mathcal{A}_{g}$ is $N_g\log N_g$, namely, $C(\mathcal{A}_g) \sim N_g\log N_g$.  
\end{theorem}
\begin{proof}
	By Lemma~\ref{efub}, for any $g \geq 0$, the resistance diameter of Apollonian network $\mathcal{A}_{g}$ is at most $\frac{5}{3}$, which together with~\eqref{commuteR}, indicates that the leading scaling for the upper bound of the cover time $C(\mathcal{A}_g)$ is $N_g\log N_g$. On the other hand, according to the result in~\cite{Fe95}, $N_g\log N_g$ is also the possible minimal scaling for the cover time $C(\mathcal{A}_g)$. Thus, we conclude that $C(\mathcal{A}_g) \sim N_g\log N_g$. 
\end{proof}
\section{Cover Time in Pseudofractal Scale-free Webs} \label{fractal}

In this section, we study the cover time in the  pseudofractal scale-free webs~\cite{DoGoMe02}. We will show that the behavior of the cover time is similar to that of the Apollonian networks.

The pseudofractal scale-free webs are also built iteratively. 
Let $\mathcal{F}_g$ denote the network after $g$ ($g \geq 0$) iterations. Initially $g=0$, $\mathcal{F}_0$ includes three nodes and three edges, forming a triangle. 
For every $g \geq 0$, the operation from $\mathcal{F}_g$ to $\mathcal{F}_{g+1}$ is as follows: For each edge in $\mathcal{F}_{g}$, one new node is generated and linked to both end nodes of this edge. 
Figure~\ref{consF} illustrates the first two iterations of the network. Let $N_{g}$ and $E_{g}$ denote, respectively, the number of nodes and edges in network $\mathcal{F}_g$. For all $g \geq 0$, $N_g = 3(3^g + 1)/2$ and $E_{g} = 3^{g+1}$. Thus, the pseudofractal scale-free webs are sparse with the average node degree being $4$. 

\begin{figure}
	\begin{center}
		\includegraphics[width=0.6\linewidth]{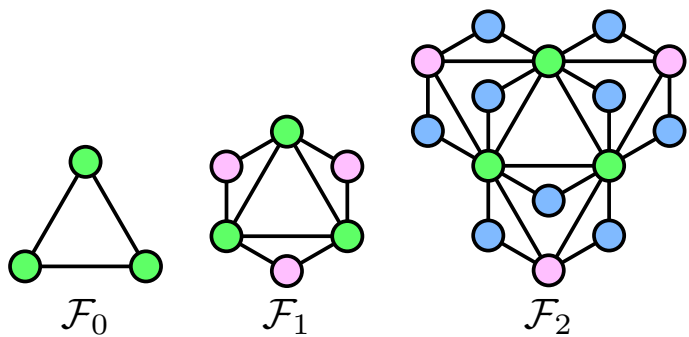}
		\caption{Construction procedure of the pseudofractal scale-free webs. }
		\label{consF}
	\end{center}
\end{figure}

The pseudofractal scale-free webs also display  the remarkable features found for many real networks~\cite{DoGoMe02}. They are scale-free with the node degree following a power-law distribution $P(d)\sim d^{-(1 + \ln 3/ \ln 2)}$. They also display the small-world effect, with their average shortest path distance scaling logarithmically with $N_g$ and their average clustering coefficient being $0.8$.

In a way similar to that of Apollonian networks, we can prove that for any pair of nodes in $\mathcal{F}_g$, its resistance distance is at most $3$, it is the same with the resistance diameter of $\mathcal{F}_g$. Based on this result, the leading behavior for the cover time $C(\mathcal{F}_g)$ of the pseudofractal scale-free web $\mathcal{F}_{g}$ is obtained, as summarized in the following theorem.
\begin{theorem}\label{coverf}
	For $g \geq 0$,  the dominating scaling of cover time $C(\mathcal{F}_g) $ for the  pseudofractal scale-free web $\mathcal{F}_{g}$ is $N_g\log N_g$, namely, $C(\mathcal{F}_g) \sim N_g\log N_g$. 
\end{theorem}


\section{Result Analysis}\label{analysis.sec}

In the previous sections, we have presented a systematic study on the cover time of random walks occurring on many real and model scale-free small-world networks. It was demonstrated that in all the networks under consideration, their cover time behaves with the node number $N$ as $N\log N$. This scaling is the same as that for the complete graph, the cover time of which is the smallest amongst all graphs with identical number of nodes. Thus, the studied networks are nearly optimal in the sense that their cover time has the minimal scaling. Since for random walks on a graph, its behavior heavily depends on the topology of the graph, we argue that the scale-free and small-world structure are responsible for the observed minimal scaling of cover time on the considered networks, which can be understood from the following heuristic explanations.

In a scale-free graph, there are a few nodes with large degree that are directly attached to many other nodes in the graph, which leads to the small-world phenomenon, with the average shortest path distance scaling at most logarithmically with the number $N$ of nodes~\cite{Ne03}. The synergy of scale-free and small-world properties strongly affects various quantities of random walks on graphs with these two features. For example, the average hitting time to a hub node behaves sublinearly with $N$~\cite{LiJuZh12}, while scales linearly to a small-degree node. In the context of cover time, as shown in~\eqref{commuteR} its upper bound is related to the resistance diameter. Recall that the resistance diameter of a graph is fully determined by the non-zero eigenvalues and their corresponding eigenvectors of its Laplacian matrix, which are in turn influenced by the structural properties of the graph. In the considered realistic and model scale-free networks, for a pair of ``remote'' nodes, there are many paths with different lengths, since they exhibit a nontrivial pattern with a number of cycles at  various scales~\cite{RoKiBoBe05,KlSt06}. As the networks grow, the resistance distance between any pair of existing nodes decreases, and the effective resistance between a new node and others is small and does not increase with the node number $N$. As a result, their resistance diameter does not depend on $N$, but converges to small constants. By~\eqref{commuteR}, their cover time is considerably small, scaling with $N$ as $N \log N$.

Note that in addition to scale-free and small-world properties, many real-world network systems also possess the community structure~\cite{GiNe02} and network motifs~\cite{MiShItKaChAl02}.  Below we show that although community structure~\cite{GiNe02} and motifs are ubiquitous in realistic systems, both of them are not necessary for the observed minimum scaling $N\log N$ of cover time. For example, the extended Sierpi{\'n}ski graphs~\cite{QiZh19} are iteratively constructed, with the complete graph  $\mathcal{K} _q$ of $q$ nodes being the basic building blocks. They thus have obvious community structure and network motifs. However, their cover time scale with network size $N$ as $N^{1+\frac{\log (q+2)}{ \log q}}\log N$, much larger than $N \log N$. Again for instance, for the modular network~\cite{FoBa07} obtained through replacing each node of a ring by the complete graph  $\mathcal{K} _q$, it is easy to verify that its  cover time behaves with $N$ as $N^2$ for large  $N$ and small $q$.

Based on the above augments, we conclude that the considered networks are almost optimal in the sense that their cover time has the minimal scaling, and we argue that their common scale-free small-world topology is responsible for the small cover time on these networks. Particularly, the two generating mechanisms, growth and preferential attachment, for the BA networks, are also the mechanisms common to a number of complex real networks, such as business networks and social networks~\cite{BaAl99}. Therefore, it is not surprising that the cover time of the studied real and model networks exhibits the same behavior, which is also consistent with our intuition.

\section{Conclusions}

The cover time is a central quantity for random walks on a graph, which has been applied to various areas, with its implication effects dependent on its behavior. It is established that amongst all graphs with identical size $N$, the complete graph is the unique optimal one possessing the minimum cover time, with the leading scaling being $N\log N$. Since complete graphs are dense, they cannot mimic realistic networks, most of which are sparse, scale-free and small-world, with small average degree, power-law degree distribution and small average distance. Thus far, the behavior of cover time for realistic networks has not been well understood. Particularly, we still lack rigorous results about cover time for model networks displaying the common properties observed for realistic systems.

In order to uncover the behavior of cover time on sparse real networks with scale-free small-world topology, in this paper, we presented an extensive empirical study on the cover time for a large variety of real scale-free small-world networks, which are abundant in computer science, physics, biology and social science. By using the link governing cover time with resistance diameter, we evaluated the cover time for the real-world networks concerned, which displays the $N\log N$ scaling with the number of nodes $N$. We also studied the cover time of three sparse  scale-free model networks: Barab{\'a}si-Albert networks, Apollonian networks,  and pseudofractal scale-free networks. For all these three networks, their cover time behaves as $N\log N$. Thus, the  minimal scaling for cover time of complete graphs can be reached in sparse scale-free small-world graphs with small constant average degree. This work enriches our understanding on the cover time in real-life networks, and provides useful insights into structure design of sparse networks with small cover time, as well as the design of algorithms related to cover time. 

It is worth mentioning that we only studied the behavior of cover time for undirected unweighted graphs by using its connection with resistance diameter.  In future work, we plan to devise efficient and effective algorithm for evaluating resistance diameter, in order to quickly determine the upper bound of cover time. On the other hand, since most real-world networks are directed and weighted, such as WWW and citation networks, future work should also include exploring the scalings of cover time in directed weighted graphs.

\bibliographystyle{ACM-Reference-Format}
\bibliography{main}

\end{document}